\documentclass[a4paper,UKenglish]{article}

\usepackage{microtype}
\usepackage{fullpage}
\usepackage{graphics}
\graphicspath{{./figures/},{./figures/}}

\usepackage{authblk}

\usepackage{graphicx}
\usepackage{amsthm}
\usepackage{amsmath}
\usepackage{enumerate}
\usepackage{color}
\usepackage{xspace}
\usepackage{url}
\usepackage{microtype}
\usepackage{graphicx}
\usepackage{subfigure}
\usepackage{booktabs} 
\usepackage[table,xcdraw]{xcolor}

\newtheorem{remark}{Remark}

\usepackage{hyperref}
\usepackage{verbatim}

\usepackage{amsfonts}\usepackage{amssymb}
\usepackage{thmtools}\usepackage{mathscinet}
\usepackage{thm-restate}
    \newtheorem{theorem}{Theorem}
    
    \newtheorem{lemma}[theorem]{Lemma}

    	\definecolor{darkgreen}{rgb}{0.01, 0.93, 0.29}
\definecolor{lightbrown}{rgb}{0.91, 0.4, 0.11}
\usepackage{framed}

\title{Positive Planar Satisfiability Problems \\under   3-Connectivity Constraints }

\author[1,3]{Md.\ Manzurul Hasan}
\author[2]{Debajyoti Mondal} 
\author[1]{Md.\ Saidur Rahman}

\affil[1]{Graph Drawing \& Information Visualization Laboratory, Department of Computer Science and Engineering,  
Bangladesh University of Engineering and Technology (BUET), Bangladesh }  
\affil[2]{Department of Computer Science, University of Saskatchewan, Saskatoon, Canada }  
\affil[3]{Department of Computer Science, 
American International University-Bangladesh (AIUB), Bangladesh 

\texttt{mhasan.cse00@gmail.com,  dmondal@cs.usask.ca,saidurrahman@cse.buet.ac.bd}}

\usepackage[textsize=tiny]{todonotes}
\usepackage{verbatim}

\begin{document}
\maketitle              
\begin{abstract}
A 3-SAT problem is called positive and planar if all the literals are positive and the clause-variable  incidence graph (i.e., SAT graph) is planar. The NAE 3-SAT and 1-in-3-SAT are two variants of 3-SAT that remain NP-complete even when they are positive. The positive  1-in-3-SAT problem remains NP-complete under planarity constraint, but planar NAE 3-SAT is solvable in $O(n^{1.5}\log n)$ time. In this paper we prove that  a positive planar NAE 3-SAT is always satisfiable when the underlying SAT graph is 3-connected, and a satisfiable assignment can be obtained in linear time. We also show that without 3-connectivity constraint, existence of  a linear-time algorithm for positive planar NAE 3-SAT problem is unlikely as it would imply a linear-time algorithm for finding a spanning 2-matching in a planar subcubic graph. We then prove that positive planar 1-in-3-SAT remains NP-complete under the 3-connectivity constraint, even when each variable appears in at most 4 clauses. However, we show that the 3-connected  planar 1-in-3-SAT is always satisfiable when each variable  appears in an even number of clauses.
\end{abstract}

\section{Introduction}
\label{sec:intro}
Boolean algebra is widely used in digital logic design to represent and simplify the Boolean operations. Possible values of the variables are true or $1$, and false or $0$. The negation or NOT operation is denoted
by $\neg{}$ or $\bar{}$. A \emph{SAT} is a Boolean formula consisting of conjunction of clauses, e.g., $\Phi = (x_1 \vee x_2 \vee \neg x_3) \wedge (x_2 \vee \neg x_5)$. An assignment for a Boolean formula is a mapping of values to its variables. With such a mapping the formula can be evaluated according to the respective rules. If the formula is satisfied, i.e. evaluates to $1$, then the assignment is called satisfying and otherwise, unsatisfying. Cook~\cite{Cook71} showed that the satisfiability problem for Boolean formulas, SAT, is NP-complete. From then on SAT has been reduced 
to many other NP problems to prove them as NP-complete. We refer the reader to~\cite{GJ79} for more details on NP-completeness.

A \emph{$3$-SAT problem} is a SAT problem where every clause contains  at most $3$-literals. A \emph{SAT graph}  $G(\Phi)$ of a $3$-SAT instance $\Phi$   consists of a vertex for each  clause and a vertex for each variable, where there exists an edge between a clause vertex  and a variable vertex if and only if the variable or its negation   appears in that clause (e.g., see Figure~\ref{figure:planar_3-sat}(a)). A $3$-SAT problem is called \emph{planar} if its  SAT graph is planar. Lichtenstein~\cite{Lic82} showed that the planar $3$-SAT problem is NP-complete.

\begin{figure}[!h]
\centering
\includegraphics[width=.75\textwidth]{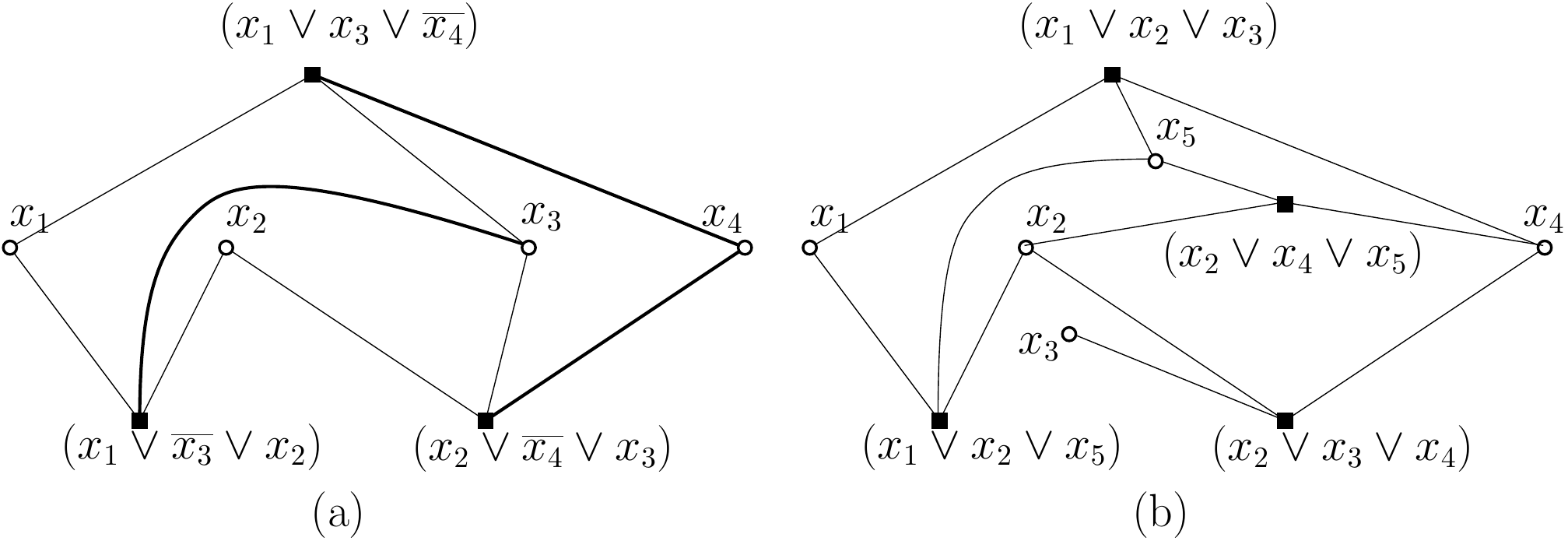}
\caption{(a) A SAT graph corresponding to the Boolean formula $\phi = (x_1 \vee x_3 \vee \overline{x_4}) \wedge (x_1 \vee x_2\vee\overline{x_3})
\wedge (x_2 \vee x_3\vee\overline{x_4}) $. The inverted variables are shown in bold edges. (b) A SAT graph corresponding to a positive planar  3-SAT. }\label{figure:planar_3-sat}
\end{figure}



A rich body of research investigates variants of 3-SAT problems~\cite{Tippenhauer16,Filho19}, and also under various restrictions on the SAT graph, e.g., when the SAT graph is planar, or 3-connected, or of bounded degree~\cite{DBLP:journals/dam/Kratochvil94}. Not-All-Equal (NAE) $3$-SAT and 1-in-3-SAT are two well-studied variants for 3-SAT. 
 In a \emph{NAE 3-SAT problem}, the goal is to find a truth assignment to the variables such that each clause contains at least one  true Boolean value and at least one false Boolean value.  In a \emph{1-in-3-SAT problem}, the goal is to find a truth assignment to the variables such that each clause contains exactly one  true Boolean value.

Both NAE 3-SAT and 1-in-3-SAT remains NP-complete even when restricted to  \emph{positive SAT}, i.e., when all the literals are positive (e.g., see Figure~\ref{figure:planar_3-sat}(b)). Planar NAE 3-SAT is known to be polynomial-time solvable. Moret~\cite{Moret88} showed that planar NAE $3$-SAT 
is in $P$, but no tight worst case time-complexity was calculated in the paper~\cite{Moret88}. Moret's idea was based on finding a min-cut in a planar graph, and thus the planar NAE 3-SAT can be solved in $O(n^{1.5}\log n)$ time~\cite{DBLP:journals/tc/ShihWK90}. The positive 1-in-3-SAT problem remains NP-complete even under  stringent conditions, i.e., when the SAT graph is planar and cubic~\cite{DBLP:journals/dcg/MooreR01}. A \emph{cubic} (resp., \emph{subcubic}) graph is a graph where the degree of each vertex is exactly  (resp., at most) three.  


A natural question in this context is to ask whether there are nontrivial variants of the positive planar NAE 3-SAT or positive planar 1-in-3-SAT that can be solved faster. In this paper we consider the 3-connectivity constraints on the SAT graph. A rich body of research examines NP-complete graph problems under various connectivity constraints~\cite{DBLP:conf/icalp/Biedl14,DBLP:conf/walcom/DurocherM12,DBLP:journals/siamcomp/GareyJT76}. The planar 3-SAT problem remains NP-hard even when the SAT graph is 3-connected and each variable appears in at most 4 clauses~\cite{DBLP:journals/dam/Kratochvil94}. However, 3-SAT is trivially satisfiable when the SAT graph is 3-connected and cubic (i.e., with only degree 3 vertices).



\paragraph{Contributions:} In this paper we examine   positive planar satisfiability problems under 3-connectivity constraints. Our contributions are as follows.

\begin{enumerate}
    \item We prove that positive planar NAE $3$-SAT  is always satisfiable when the SAT graph is  $3$-connected, and a satisfiable assignment can be obtained in linear time. \item We show that without 3-connectivity constraint, the positive planar NAE 3-SAT problem is as hard as finding a spanning 2-matching in a planar cubic graph (i.e., a planar graph with only degree 3 vertices). A \emph{spanning 2-matching} of a graph is a spanning subgraph with maximum degree 2. Since no linear-time algorithm is  known for finding a spanning 2-matching  in a planar cubic graph, finding a linear time algorithm  for NAE 3-SAT appears to be challenging.
    \item We prove that positive planar 1-in-$3$-SAT  remains NP-complete even under   $3$-connectivity constraint and when every variable appears in  at most 4 clauses.

    \item In contrast, we show that positive planar 1-in-$3$-SAT is always satisfiable when every variable appears in  an even number of clauses, and a satisfiable assignment can be obtained in quadratic time.
\end{enumerate}

The rest of the paper is organized as follows. Section~\ref{sec:ppc} shows that positive planar NAE 3-SAT is always satisfiable and provides a linear-time algorithm to compute such  a satisfiable assignment.  Section~\ref{sec:lb} reduces spanning 2-matching to positive planar NAE 3-SAT. Section~\ref{sec:hard} proves the NP-hardness of positive planar 1-in-3-SAT   even when the SAT graph is 3-connected and  every variable appears in at most 4 clauses. Section~\ref{sec:even} proves that positive planar 1-in-$3$-SAT is always satisfiable when every variable appears in  an even number of clauses. Finally, Section~\ref{sec:con} concludes the paper suggesting directions for future research.

\section{Preliminaries}
In this section we give some definitions that will be used throughout the paper and present some preliminary results.

A plane graph is a planar graph with a fixed planar embedding in the plane. A planar graph may have an exponential number of embeddings. A plane graph splits the plane into connected regions called \emph{faces}. The unbounded region is called the \emph{outer face} and the other regions are called \emph{inner faces}. The vertices that lie on the unbounded face are called \emph{outer vertices} and the remaining vertices are called \emph{inner vertices}.

A graph is bipartite if and only if it is bichromatic, i.e. the graph's vertices can be colored  with at most two colors such that no two adjacent vertices get the same color. The \emph{connectivity} $\kappa(G)$\index{$\kappa(G)$} of
a graph $G$ is the minimum number of vertices whose removal results in a disconnected graph or a single-vertex graph. By Menger's theorem, every pair of vertices $u,v$ in a  $k$-connected graph has at least $k$ vertex-disjoint paths (except for the common vertices $u,v$)  connecting $u$ and $v$. A plane graph is \emph{internally $k$-connected} if for every inner vertex $w$, there are $k$ vertex-disjoint paths (except for the common vertex $w$) that start  at $w$ and end  at an outer vertex.  We refer the reader to~\cite{Rahman17} for basic terminologies on graphs.

\begin{lemma}\label{3con}
Let $G$ be a 3-connected plane graph with a vertex $v$ of degree $d\ge 3$, where $w_1,\ldots, w_d$ are the neighbors of $v$.  Let $H$ be a 2-connected and internally 3-connected plane graph with at most $d$ outer vertices of degree two. Let $G'$ be a graph obtained by replacing $v$ with $H$ and connecting $w_1,\ldots, w_d$ to at least three outer vertices of $H$ such that the graph remains planar and  every degree-two outer vertex of $H$ obtains a new edge. Then $G'$ is 3-connected.
\end{lemma}
\begin{proof}
Assume for a contradiction that $G'$ is not 3-connected and let $u,v$ be a pair of vertices such that deleting them generates a disconnected graph. We now show that such a pair cannot exist in $G'$.

Let $S$ be the set of vertices in $G'$ that correspond to the inner vertices  of $H$. Let $G''$ be the graph obtained by removing the vertices in $S$ from $G'$. Since $H$ is 2-connected, $G''$ can be seen as a graph  obtained from $G$ by replacing $v$ with a cycle where the neighbors of $v$ are connected to at least 3 distinct neighbors on the cycle. Such graphs are known to be 3-connected~\cite{mastersthesis}. Therefore, either both $u,v$ lie in $S$, or exactly one of them must lie in $S$.

First consider the case when both $u,v$ lie in $S$. Since deleting $u,v$ generates a disconnected graph, there must be a connected component $C$ that belongs to $H$. Let $w$ be a vertex in $C$. Then there cannot exist 3 vertex disjoint paths from $w$ to the outer face of $H$, which contradicts that $H$ is internally 3-connected.

Consider now that exactly one of $u$ and $v$   lies in $S$. Without loss of generality assume that $u$ belongs to $S$. Since $G''$ is 3-connected, deleting $v$ does not disconnect $G''$. Since $H$ is 2-connected, deleting $u$ does not disconnect $H$. By the construction there are three disjoint edges connecting the neighbors of $v$ and $H$. Hence deleting $u$ and $v$ cannot disconnect $G'$. 
\end{proof}


Let $\Phi$ be a positive planar  $3$-SAT and let $G$ be its corresponding SAT graph. Let $\Gamma$ be any arbitrary planar embedding of $G$. We call 
$\Gamma$ a ~\emph{quadrangulated SAT graph} if every face of $\Gamma$ has exactly four vertices, where two of them are clause vertices and two are variable vertices. A   \emph{clause graph} of a quadrangulated SAT graph $\Gamma$ is obtained by adding for every face, an edge between its  clause vertices, and finally, removing the variable vertices. 

\begin{figure}[h]
\centering
\includegraphics[width=.85\textwidth]{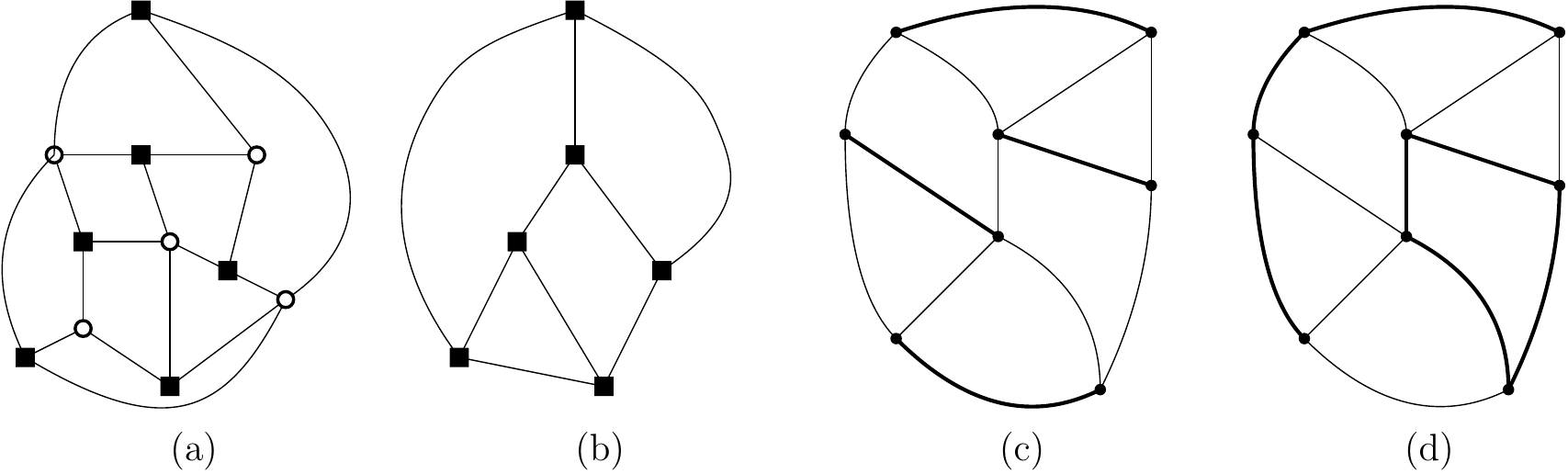}
\caption{(a) A quadrangulated positive planar   $3$-SAT graph $\Gamma$, where the clause vertices are shown in black square and the variable vertices are shown in circles. (b) 
The corresponding clause graph $C$ of $\Gamma$. (c) A  perfect matching and (d) a spanning 2-matching, where the edges in the matching and 2-matching are shown in bold. }
\label{figure:matching}
\end{figure}

A \emph{perfect matching} of a graph is a collection of edges $M$ such that every vertex is incident to exactly one edge in $M$, as illustrated in Figure~\ref{figure:matching}(c). A {cubic (resp. subcubic)} graph where every vertex is of degree exactly 3 (resp., at most 3). By Petersen's theorem~\cite{Petersen}, every bridgeless cubic graph has a perfect matching, and such a matching can be found efficiently under planarity constraint.  
\begin{lemma}[Biedl et al.~\cite{Biedl01}]
\label{lem:biedl}
A perfect matching in a planar  bridgeless cubic 
graph with $n$ vertices can be found in $O(n)$ time.
\end{lemma}

A \emph{$2$-matching} of a graph $G$ is a subgraph $H$ with maximum degree two. A 2-matching  $H$ is called  \emph{spanning} if every vertex of $G$  is incident to at least one edge of $H$, as shown in Figure~\ref{figure:matching}(d). A maximum cardinality 2-matching is a 2-matching that can be computed in $O(n^{1.5})$ time~\cite{DBLP:journals/siamjo/HartvigsenL11}.

Let $\Gamma$ be a planar embedding of a set of disjoint cycles. Then a \emph{genealogical tree} $T$ of $\Gamma$  is defined as follows: 
\begin{enumerate}[-]
    \item Each vertex $v$ in $T$ corresponds to a face $f_v$ in $\Gamma$. 
    \item The root of $T$ corresponds to the outerface of $\Gamma$.
    \item There exists an edge from a parent node $v$ to a child node $w$ if an only if the face $f_v$  encloses the face $f_w$, and $f_v$ and $f_w$ share  a common cycle on their boundaries.  
\end{enumerate}
\section{Positive Planar $3$-Connected NAE $3$-SAT}
\label{sec:ppc}
In this section we show that a positive planar $3$-connected NAE $3$-SAT is always satisfiable and a satisfying assignment can be obtained in $O(n+m)$ time, where $n$ and $m$ are the number of variables and clauses, respectively.  

\begin{theorem}
\label{th:nae3sat}
Let $R$ be an  arbitrary positive planar $3$-connected NAE $3$-SAT expression. Then $R$ is always satisfiable and a satisfiable assignment of $R$ can be computed in linear time. 
\end{theorem}

Since the SAT graph $G$ is 3-connected, it has a unique plane embedding (upto the choice of the outerface) and hence the clause graph is the same for every choice of the outerface. Before we prove Theorem~\ref{th:nae3sat}, we consider a simpler case when the SAT graph $G$ is planar, 3-connected and quadrangulated, as stated   below.



\begin{lemma}
\label{lemma:pln_another}
 Let $R$ be any arbitrary positive planar $3$-connected NAE $3$-SAT instance and let $G$ be the corresponding planar SAT graph. If $G$ is quadrangulated, then $R$ is always satisfiable.
\end{lemma} 
\begin{proof}
We first show that the clause graph corresponding to $G$ must have a perfect matching and we use that matching to find a satisfying truth assignment. Since $G$ is quadrangulated, every face in $G$ has exactly four vertices, where two are variable vertices and  two are clause vertices. Let $C$ be the clause graph obtained from $G$. Since each clause  has exactly  three literals, the corresponding clause vertex $v$  has exactly 3 neighbors in $G$. Since $v$ is incident to exactly 3 faces in $G$, it must have exactly three neighbors in the clause graphs. Therefore, the clause graph $C$ is a cubic graph. We now show that $C$ is a bridgeless  cubic graph. Suppose for a contradiction that $C$  has a bridge $(v,w)$ and deleting the bridge results into two disjoint connected components $H_1$ and $H_2$ (e.g. see Figure~\ref{figure:clause-tree}(a)). Since the faces of $G$ are quadrangulated, there must be a face $p,v,q,w$ in $G$, where $p,q$ are variable vertices. Since $G$ is 3-connected,  there must be three vertex disjoint paths between $v$ and $w$. Hence there exists  a path $v,\ldots,w$ in $G$ that does not pass  through $p$ or $q$  (e.g. see Figure~\ref{figure:clause-tree}(b)). The sequence of clause vertices in this path connects $v$ and $w$ in $C$  (e.g. see Figure~\ref{figure:clause-tree}(c)). Therefore, $(v,w)$ cannot be a bridge in $C$. By Petersen's theorem~\cite{Petersen}, $C$ contains a perfect matching $M$.

\begin{figure}[h]
\centering
\includegraphics[width=.55\textwidth]{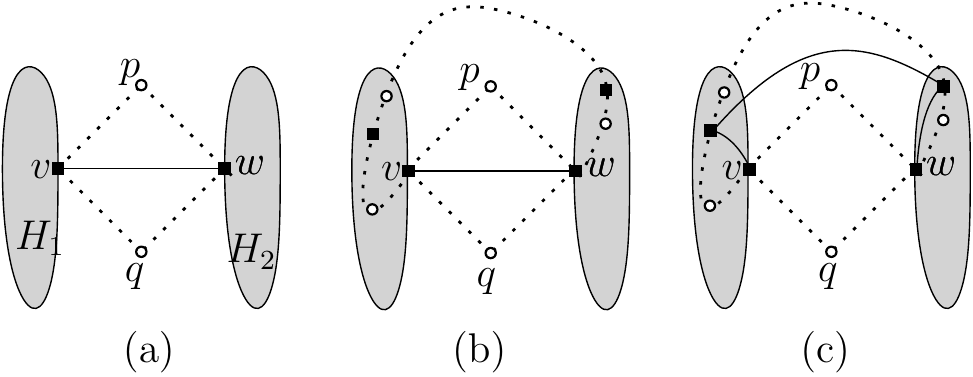}
\caption{Illustration for the proof of  Lemma~\ref{lemma:pln_another}. (a) A bridge. (b) The corresponding face and a $v$ to $w$ path that does not go through $p$ and $q$. (c) A path connecting $v$ and $w$ after deleting the edge $(v,w)$.}
\label{figure:clause-tree}
\end{figure}



We now show how to compute a satisfying assignment for the NAE 3-SAT instance $R$. Let $C'$ be the graph obtained from $C$ by deleting the edges of $M$, e.g. see Figures~\ref{figure:clause-tree2}(a)--(b).  Since $C$ is a planar  cubic graph, $C'$ must be a planar disjoint collection of cycles. Let $T$ be the genealogical tree of $C'$. We compute a two coloring of $T$ with red and black colors, e.g. see Figure~\ref{figure:clause-tree2}(c). For each vertex  (i.e., face) which has been colored red, we set the corresponding variable vertices (i.e., the variable vertices lying inside  face) to be true~\ref{figure:clause-tree2}(d).

\begin{figure}[h]
\centering
\includegraphics[width=1\textwidth]{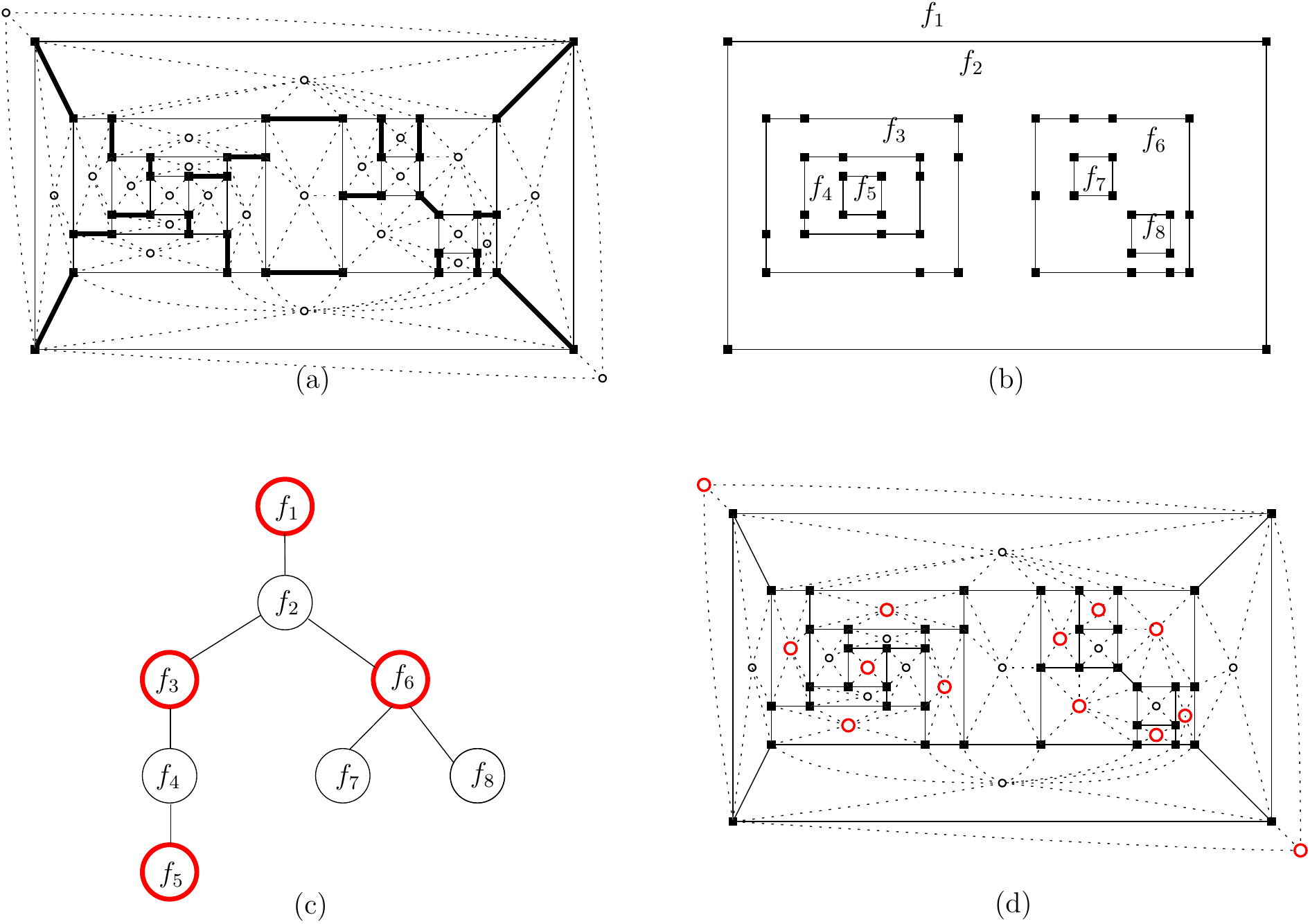}
\caption{Illustration for computing a satisfying truth assignment from a perfect matching. (a) A quadrangulated SAT graph, where clause graph $C$ is shown in solid lines. A perfect matching $M$ is shown in bold. (b) The graph $C'$. (c) A genealogical tree $T$, with a two coloring where  $\{f_1,f_3,f_6,f_5\}$ are colored with the same color. (d) A satisfying  truth assignment obtained from the  two coloring of $T$.}
\label{figure:clause-tree2}
\end{figure}

We now prove that the resulting truth assignment is a satisfying truth assignment, i.e., for each clause at least one variable vertex must be true and at least one must be false. Consider a clause vertex $q$ in $G$. Note that $q$ is incident to three variable vertices  $v_1,v_2$ and $v_3$ in $G$, and let $D$ be the cycle passing through $q$ in  $C'$. Then  $D$ either contains two variable vertices in its interior and the other variable vertex remains  outside, or $D$ contains one variable vertex in its interior and the other two remain   outside. Since the genealogical tree is two colored, at most two of these variable vertices of $q$ can be true, and the remaining ones must be false. 
\end{proof}

\begin{figure}[h]
\centering
\includegraphics[width=.8\textwidth]{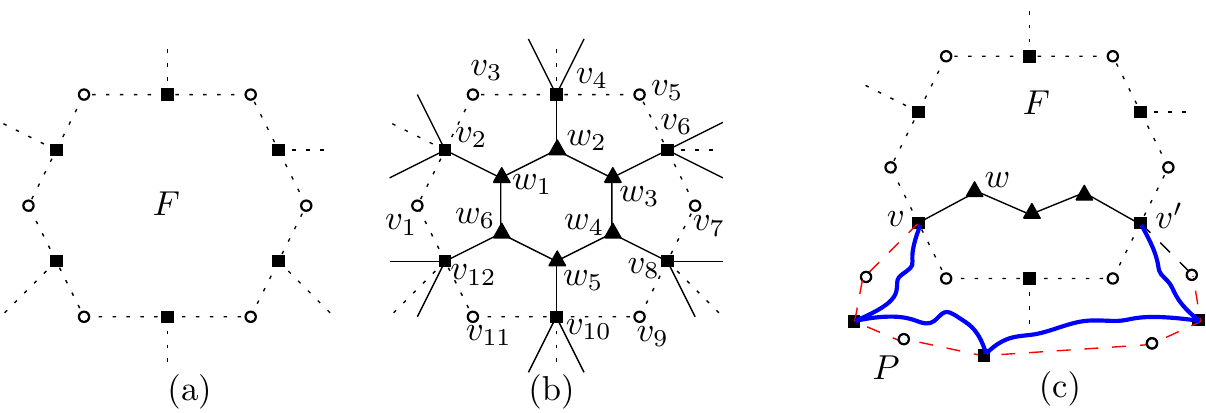}
\caption{Illustration for the proof of  Lemma~\ref{lemma:pln_another}. (a) A face of length 12 in $G$.  (b) Illustration for the saturation operation, where the clause vertices are shown in squares and the vertices of the added cycle are shown in triangles. (c) The vertices $v$ and $w$ lie on a cycle in $C_s$.}
\label{figure:saturation}
\end{figure}

We now consider the case when the SAT graph $G$ is not necessarily quadrangulated. If the SAT graph is not quadrangulated, then the clause graph may not be cubic or planar. Therefore, instead of using the clause graph, we define a \emph{saturated clause graph $C_S$}    as follows. 

\textbf{Saturated clause graph $C_s$:} Let $\Gamma$ be a plane embedding of $G$. Let $F$ be a face $v_1,v_2,\ldots, v_p$ of $\Gamma$ with $p>4$ vertices. Since $G$ is 3-connected and bipartite, $p$ must be even and exactly half of the vertices would be clause vertices. Let $v_2,v_4,\ldots,v_{p}$ be the clause vertices. We define a \emph{saturation operation} that first  adds a cycle $w_1,w_2,\ldots, w_{p/2}$ of $p/2$ dummy   vertices interior to $F$ and  then adds the edges  $(w_i,v_{2i})$.   Figures~\ref{figure:saturation}(a)--(b) illustrate the saturation operation. The saturated clause graph is obtained by adding for each quadrangular face an edge between the clause vertices, and then  applying the saturation operation to all the faces of length more than four in $\Gamma$.

We are now ready to prove Theorem~\ref{th:nae3sat}.
\smallskip



\noindent
\textbf{Proof of Theorem~\ref{th:nae3sat}}. Let $R$ be any arbitrary positive planar $3$-connected NAE $3$-SAT expression and let $G$ be the corresponding planar 3-connected SAT  graph. Let $C_s$ be the saturated clause graph of $G$. Since $G$ is 3-connected, each clause vertex in $G$ is adjacent to exactly three variable vertices. Therefore, it is straightforward to observe from the  construction of saturated clause graph that $C_s$ is planar and cubic. We now show that $C_s$ is bridgeless. 

Suppose for a contradiction that $(v,w)$ is a bridge in $C_s$. If $(v,w)$ is an edge inside   a quadrangular face of $G$, then we can prove that there must be another path connecting $v$ and $w$ in $C_s$ in the same way as we proved the clause graph to be bridgeless in Lemma~\ref{lemma:pln_another}.
 If $(v,w)$ is an edge that has been added during the saturation operation on some face $F$, then both $v$ and $w$ cannot be on the added cycle. Therefore, we may assume without loss of generality that $v$ is a clause vertex (Figure~\ref{figure:saturation}(c)) and $w$ is a dummy vertex. Let $v'$ be another clause vertex on $F$. Since $G$ is 3-connected, there must be a path $P$  in $G$ between $v$ and $v'$ that does not contain any vertex of $F$. Hence we can construct a path in $C_s$ between $v$ and $v'$ outside of $F$, and extend it inside $F$ to form a cycle that contains $(v,w)$. Hence $(v,w)$ cannot be a bridge in $C_s$.
 
Since $C_s$ is planar bridgeless cubic graph, by Petersen's theorem~\cite{Petersen} $C_s$ contains a perfect matching. We can now use  the same argument as in the proof of  Lemma~\ref{lemma:pln_another} using this perfect matching to construct a satisfying truth assignment of $R$.

It now remains to prove that the time complexity of the whole process is linear in the number of vertices of $G$. A planar embedding $\Gamma$ of the SAT graph $G$ can be obtained in linear time. The construction of $C_s$ requires iterating through each face of $\Gamma$ and spending a time proportional to the length of each face. Hence we can compute $C_s$ in linear time. Since $C_s$ is a planar bridgeless cubic graph, by  Lemma~\ref{lem:biedl}, one can obtain a perfect matching $M$ of $C_s$ in linear time. Given a perfect matching, one can delete the edges of $M$  from $\Gamma$ and then recursively traverse the cycles on the outer face to construct the genealogical tree $T$. Thus the construction of the $T$ takes linear time. Finally, coloring the tree with two colors and setting the corresponding truth values takes a linear-time traversal of the tree and a linear-time traversal of $G$. Thus the overall time complexity remains linear.


\section{Positive Planar  NAE $3$-SAT without 3-Connectivity Constraint}
\label{sec:lb}
In this section we consider the case of general  planar SAT graphs. We show that the problem of solving a positive planar  NAE 3-SAT is as hard as the problem of deciding whether a  planar cubic graph contains a spanning 2-matching. Figure~\ref{figure:eq}(a) illustrates a spanning 2-matching in a cubic graph. Although a rich body of literature examines 2-factor and maximum 2-matching in cubic graphs~\cite{DBLP:journals/disopt/Kobayashi10}, to the best of our knowledge, no linear-time algorithm is known for deciding whether a planar cubic graph admits a spanning 2-matching. 

\begin{theorem}
The problem of deciding whether  a  connected planar cubic graph admits a spanning 2-matching is   linear-time reducible to  positive planar   NAE 3-SAT.
\end{theorem}

\begin{proof}
Let $G$ be a planar cubic graph. 
We construct a graph $G'$ by subdividing each 
edge of $G$ with a division vertex, i.e., each edge $(u,v)$ of $G$ is replaced by a path $u,d_{uv},v$ in $G'$, where $d_{uv}$ is the division vertex. 

We now consider $G'$ as a SAT graph where the original vertices of $G$ are the clause vertices  and the division vertices are the variable vertices. In the following we show that $G$ has a spanning 2-matching if and only if the planar   NAE 3-SAT $\mathcal{I}$ corresponding to $G'$ has an affirmative not-all-equal solution.

\begin{figure}[h]
\centering
\includegraphics[width=.8\textwidth]{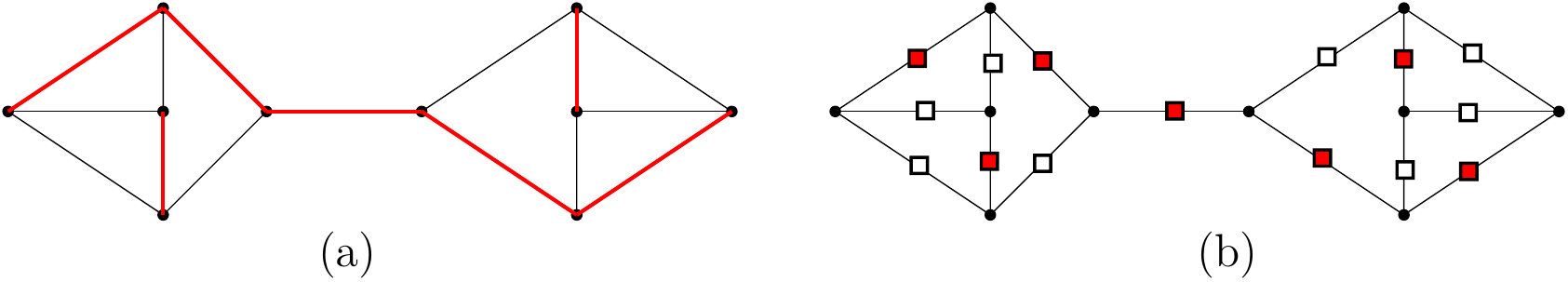}
\caption{(a) A planar cubic graph $G$ with a spanning 2-matching (shown in red).   (b) The corresponding positive planar NAE 3-SAT, and the associated affirmative solution. The literals which are true are shown in red.  }\label{figure:eq}
\end{figure}

First assume that $G$ has a spanning 2-matching (i.e., a spanning subgraph with maximum degree 2), and let $F$ be the set of edges in that  spanning subgraph. Let $F_d$ be  the division vertices in $G'$ that correspond to $F$.  We now set  the literals of $\mathcal{I}$ determined by $F_d$ to be true, and the remaining literals to false. Suppose for a contradiction that such an assignment gives rise to a clause with all true or all false literals. If all three are false, then it contradicts that $F$ is spanning. If all three are true, then there must be a vertex in $G$ that is incident to three edges in $F$, contradicting that $F$ corresponds to a subgraph with maximum degree 2.

Assume now that $\mathcal{I}$ has an affirmative not-all-equal solution. We construct a set $D$ of division vertices by taking for each clause vertex, the division vertices corresponding to the literals which are true. Let $F$ be the edges of $G$ corresponding to the set $D$. 


We now show that that $F$ corresponds to a spanning 2-matching of $G$. Since every clause contributed to $D$, $F$ must be a spanning subgraph. Suppose for a contradiction that the graph determined by $F$ contains a vertex of degree 3. By construction of $F$, this implies that there is a clause where all its literals are assigned the value true, which contradicts that the initial assignment is a not-all-equal assignment.
%

\end{proof} 

\section{Positive Planar $3$-Connected $1$-in-$3$-SAT}
\label{sec:hard}
The 1-in-3-SAT problem remains NP-hard even if each variable appears in at most 4 clauses~\cite{DBLP:journals/corr/abs-1802-09465}. Laroche~\cite{Laroche} proved that the  positive planar  1-in-3-SAT problem is NP-complete.   Moore and 
Robson~\cite{DBLP:journals/dcg/MooreR01}   proved that the  positive planar  1-in-3-SAT remains NP-hard even when the SAT graph is cubic. None of these hardness reductions ensures the 3-connectivity of the underlying SAT graph. However,  the problem remains hard even in cases where   additional edges can be added to the SAT graph (keeping it planar) to form a cycle containing all variable and clause vertices~\cite{DBLP:journals/dmtcs/Pilz19}. Although such an edge  augmented graph may be 3-connected, the SAT graph itself may not be 3-connected. 

In this section we show that positive planar 1-in-3-SAT remains NP-hard even when the SAT graph is 3-connected and each  variable  appears in at most 4 clauses.

\subsection{Outline of the Reduction}
We reduce planar 3-connected 3-SAT which is shown to be NP-hard even when every variable appears in at most 4 clauses~\cite{DBLP:journals/dam/Kratochvil94}. For convenience we will refer to the SAT graph  as a 3-SAT graph or a 1-in-3-SAT graph depending on the SAT instance. Let $G$ be a 3-SAT graph corresponding to a planar 3-connected 3-SAT instance $I$, as illustrated in Figure~\ref{fig:one-in-three}(a). The hardness reduction is carried out in two phases. 

In the first phase, we replace each variable with a variable gadget and each clause with a clause gadget. We will use the resulting planar graph $G'$ as a 1-in-3-SAT graph of a positive planar 1-in-3-SAT instance $I'$ and show that $I$ is satisfiable if and only if $I'$ is satisfiable. While constructing $G'$, we will ensure that every variable appears in at most 4 clauses. However, $G'$ would not be 3-connected. In the second phase, we will add additional gadgets to $G'$ to construct a planar 3-connected graph $G''$ ensuring that each variable appears in at most 4 clauses. 

To complete the proof we will use the resulting planar graph $G''$ as a 1-in-3-SAT graph of a positive planar 1-in-3-SAT instance $I''$ and show that $I$ is satisfiable if and only if $I''$ is satisfiable. 

\subsection{Construction of $G'$}
The graph $G'$ is constructed by replacing the vertices and clauses with the vertex and clause gadgets. 

\subsubsection{Clause Gadgets} We first replace each variable vertex $v$ of degree $d$ in $G$  with a cycle $L_v$ of $d$ vertices. We will refer to $L_v$ as the \emph{lamina of $v$}. We then connect the  vertices on $L_v$ with the neighbors of $v$ such that the resulting graph remains  planar, e.g., see Figures~\ref{fig:one-in-three}(a)--(b). For each clause vertex $c = (x\vee y \vee z)$, the clause gadget contains the clauses  $c_1 = (\overline{x} \vee p \vee r), c_2 =  (y\vee p\vee q)$ and $c_3 =  (\overline{z}\vee s\vee q)$.  Figure~\ref{fig:one-in-three}(c) illustrates a clause gadget in blue. It is known that $c$ evaluates to true if and only if $(c_1 \wedge c_2 \wedge c_3)$ admits a satisfiable truth assignment where each clause contains exactly one true value~\cite{DBLP:journals/corr/abs-1802-09465}. 


\begin{figure}[h]
\centering 
\includegraphics[width=.8\textwidth]{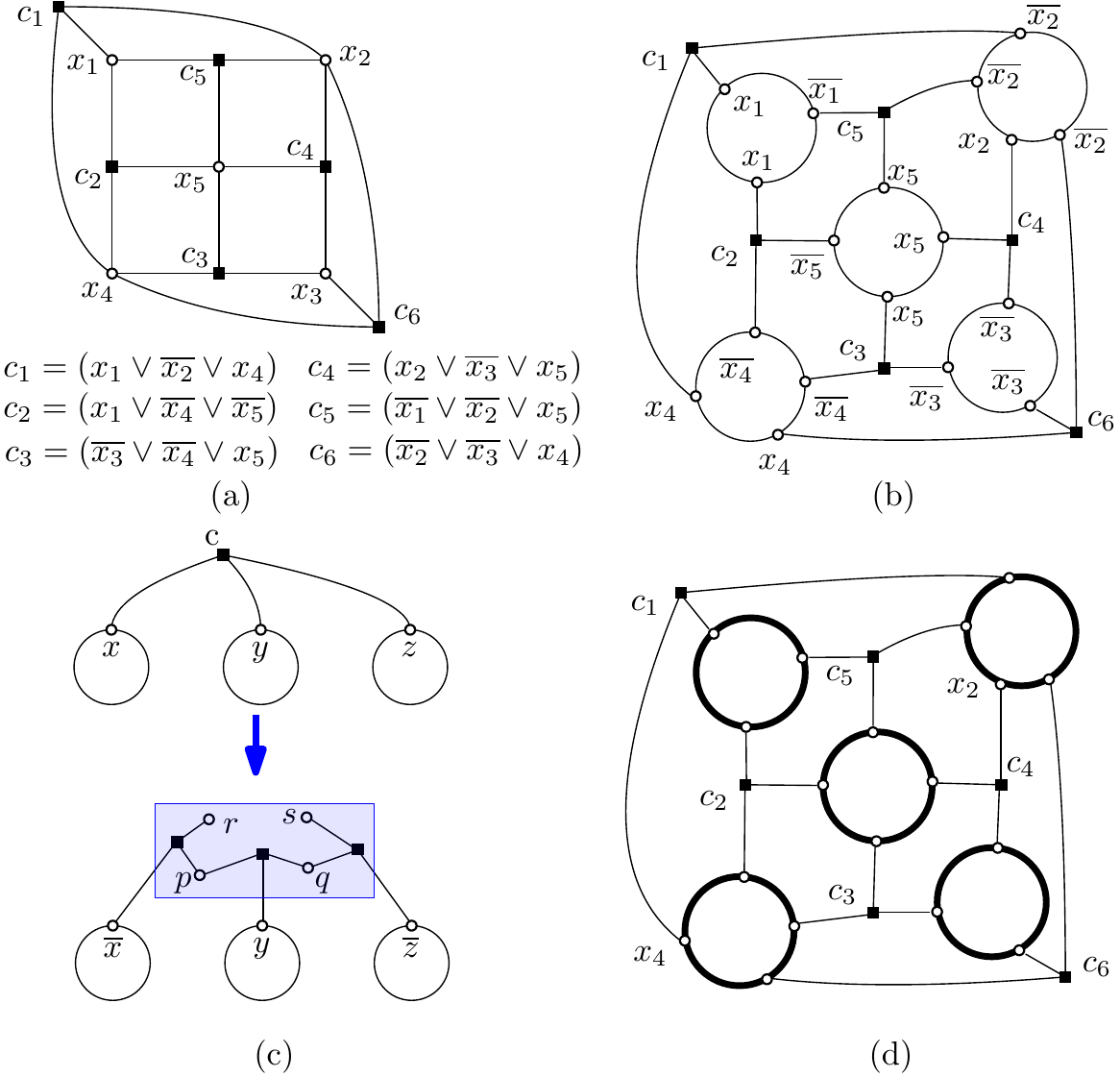}
\caption{(a) The 3-SAT graph $G$. (b) Replacing each  variable vertex with a lamina. (c)--(d) Construction of $G'$. The clause gadget is shown in a blue rectangle. The variable gadget is shown in a thick black circle.  }\label{fig:one-in-three}
\end{figure}

\subsubsection{Variable Gadgets} To remove the negated variables we now  replace each lamina with a variable gadget.  The variable gadget consists of  one inner ring, and one or more outer rings, as described below. 

A \emph{$k$-ring} is a planar 1-in-3-SAT graph of $3k$ clause vertices on its outerface, as illustrated in Figure~\ref{fig:ring}(a). Here  $k\ge 3$ is a positive integer. A $k$-ring consists of $k$ groups, each containing 6 clause vertices, as shown in red shaded region. A $(k+1)$-ring can be constructed by adding a group to a $k$-ring, as shown in Figure~\ref{fig:ring}(b). Later, we will refer to a $k$-ring  just as a \emph{ring}  for simplicity. We will use the following property of a ring. 

\begin{figure}[h]
\centering 
\includegraphics[width=\textwidth]{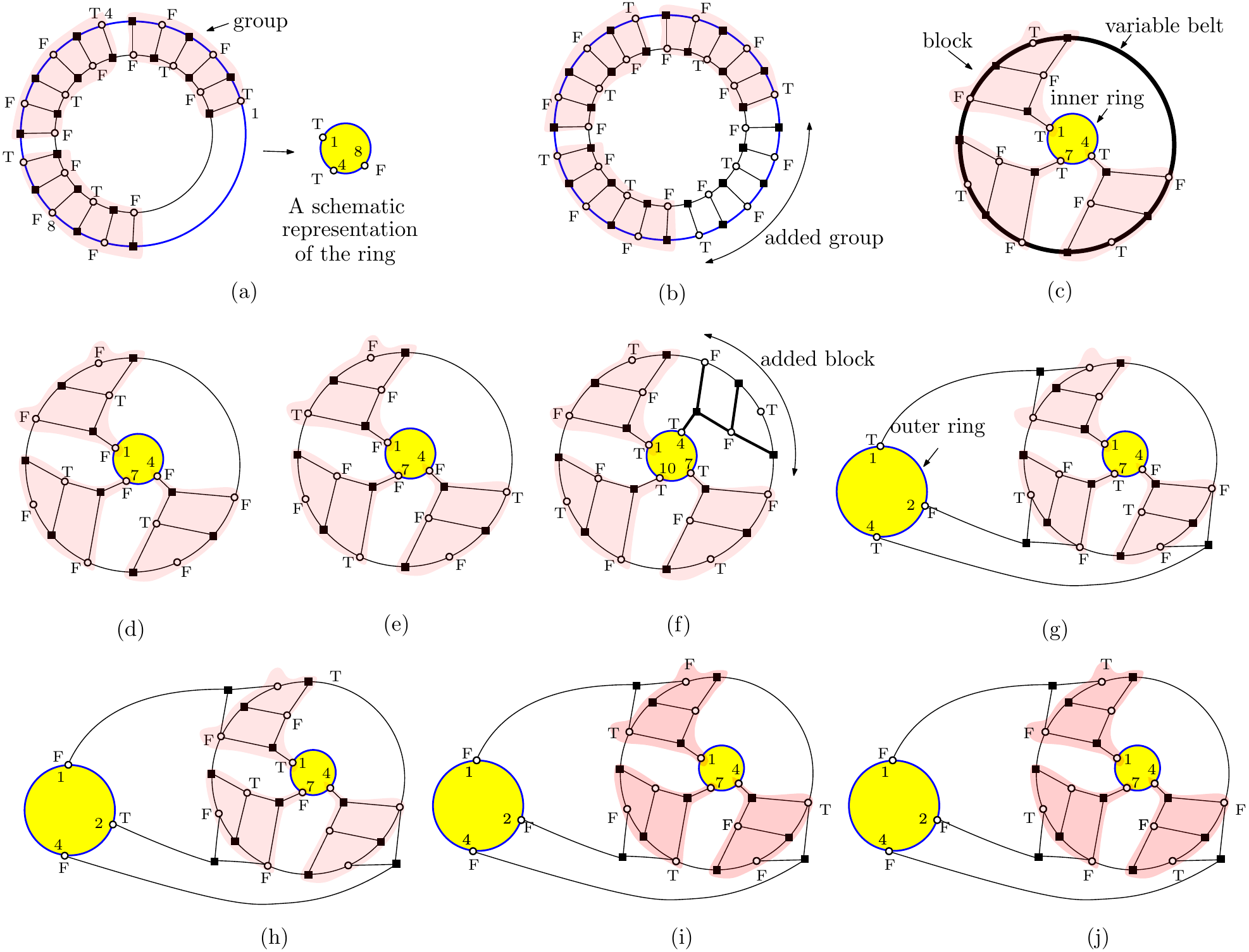}
\caption{(a) A ring with its schematic representation where the truth assignments of the variables at positions 1, 4 and 8 are shown.   The feasible truth value assignments are illustrated using T and F. The clause vertices are shown in squares. (b) An addition of a new group into a ring. (c)--(f) Illustration for variable gadgets without the outer rings. Only the outer cycle of the ring (in blue) is drawn for simplicity. (g)--(h) Variable gadgets  with   two infeasible truth value assignments. (h)--(j) Variable gadgets  with  two feasible truth value assignments.  }\label{fig:ring}
\end{figure}

\begin{remark}
\label{r1}
In every satisfiable truth assignment of a ring, the truth values on the outerface appear in the following  sequence $T,F,F,\ldots,T,F,F$.  
\end{remark}

A \emph{variable gadget of length $q$} is a planar 1-in-3-SAT graph with $q$ blocks. Here $q\ge 3$ is a positive integer and each block contains 3 clause vertices. Figure~\ref{fig:ring}(i) illustrates a variable gadget of length 3 where the three blocks are enclosed in three disjoint red shaded regions.  A variable gadget contains two or more rings: one inner ring, and one or more outer rings. A ring is  called an \emph{ outer ring} if it lies on the outer face. Otherwise, it is called an  \emph{inner ring}. The gadget in Figure~\ref{fig:ring}(i) contains two rings: one inner ring and one outer ring.  We  first describe the details of the gadget without the outer rings, as illustrated in Figure~\ref{fig:ring}(c).  We will refer to the  cycle passing through the blocks (shown in bold) as the \emph{variable band}.

One clause vertex from each block is adjacent to the inner ring.  The $i$th block, where $0\le i\le (q-1)$, is adjacent to the $(3i+1)$th variable vertex on the outerface of the inner ring. Figure~\ref{fig:ring}(c) illustrates the positions of the vertices of the inner ring where the  blocks connect to (i.e., shown in numeric labels). By Remark~\ref{r1}, these  $(3i+1)$th variables must all have the same truth value. For example, in Figure~\ref{fig:ring}(c)  they are all true, and in Figure~\ref{fig:ring}(d)--(e) they are all false. Hence, in a  satisfiable truth assignment, the truth values on the  variable band  are either all false (Figure~\ref{fig:ring}(d)), or appear in the  following  sequence $T,F,\ldots,T,F$ (Figures~\ref{fig:ring}(c) and (e)). 

An outer ring enforces the truth value sequence on the variable band to be $T,F,\ldots,T,F$. Therefore, with an outer ring, Figure~\ref{fig:ring}(d) becomes infeasible. A variable gadget of length $(q+1)$ can be constructed by adding a new  block to a variable gadget of length $q$, as shown in Figure~\ref{fig:ring}(f). Hence, we can attach outer rings as necessary.

There can be one or more outer rings exterior to the variable belt. Each outer ring is adjacent to 8 consecutive variable vertices on the belt. Figure~\ref{fig:ring}(g) illustrates the positions where the variable vertices on the outer ring form connections with the belt. By Remark~\ref{r1}, these three variable vertices on the outer ring must take a truth assignment from $\{(T,F,T),(F,T,F),(F,F,F)\}$. However, two of these assignments, i.e., $\{(T,F,T),(F,T,F)\}$ cannot be extended to a full satisfiable assignment. Figures~\ref{fig:ring}(g) and (h) illustrate this by showing how these truth assignments impose  both true and false values on the inner ring. The other truth assignment, i.e., $(F,F,F)$, enforces the  sequence $T,F,\ldots,T,F$ on the variable belt, as shown in Figures~\ref{fig:ring}(i)--(j). 

\begin{remark}
\label{r2}
In every satisfiable truth assignment of a variable gadget, the truth values on the variable belt in the following  sequence $T,F,\ldots,T,F$.  
\end{remark}

 We complete the construction of $G'$ by replacing each lamina with a variable gadget. Since every variable vertex in $G$ appears in at most 4 clauses, a lamina can have at most 4 incident edges. Figure~\ref{fig:vg} illustrates various scenarios while replacing a lamina of four variable vertices $a,b,c,d$. Since each of $a,b,c,d$ is  a positive literal (type P) or a negative literal (type N)  of the same  variable, we can have 16 different scenarios. These cases have been grouped into four configurations. While replacing a lamina  with a variable gadget, one must choose a configuration based on the literal types. Figures~\ref{fig:vg}(e)-(h) illustrate the first four scenarios of Figure~\ref{fig:vg}(b) in more detail.  For example, if $a,b,c,d$ are of type PPNP (e.g., see Figure~\ref{fig:vg}(e)), then we can choose the configuration of  Figure~\ref{fig:vg}(b).   If $a,b,c,d$ are of type PPPN (e.g., see Figure~\ref{fig:vg}(g)), then we choose the configuration of  Figure~\ref{fig:vg}(b) by   taking a mirror reflection and relabeling the variable vertices.  
 
\begin{figure}[h]
\centering 
\includegraphics[width=\textwidth]{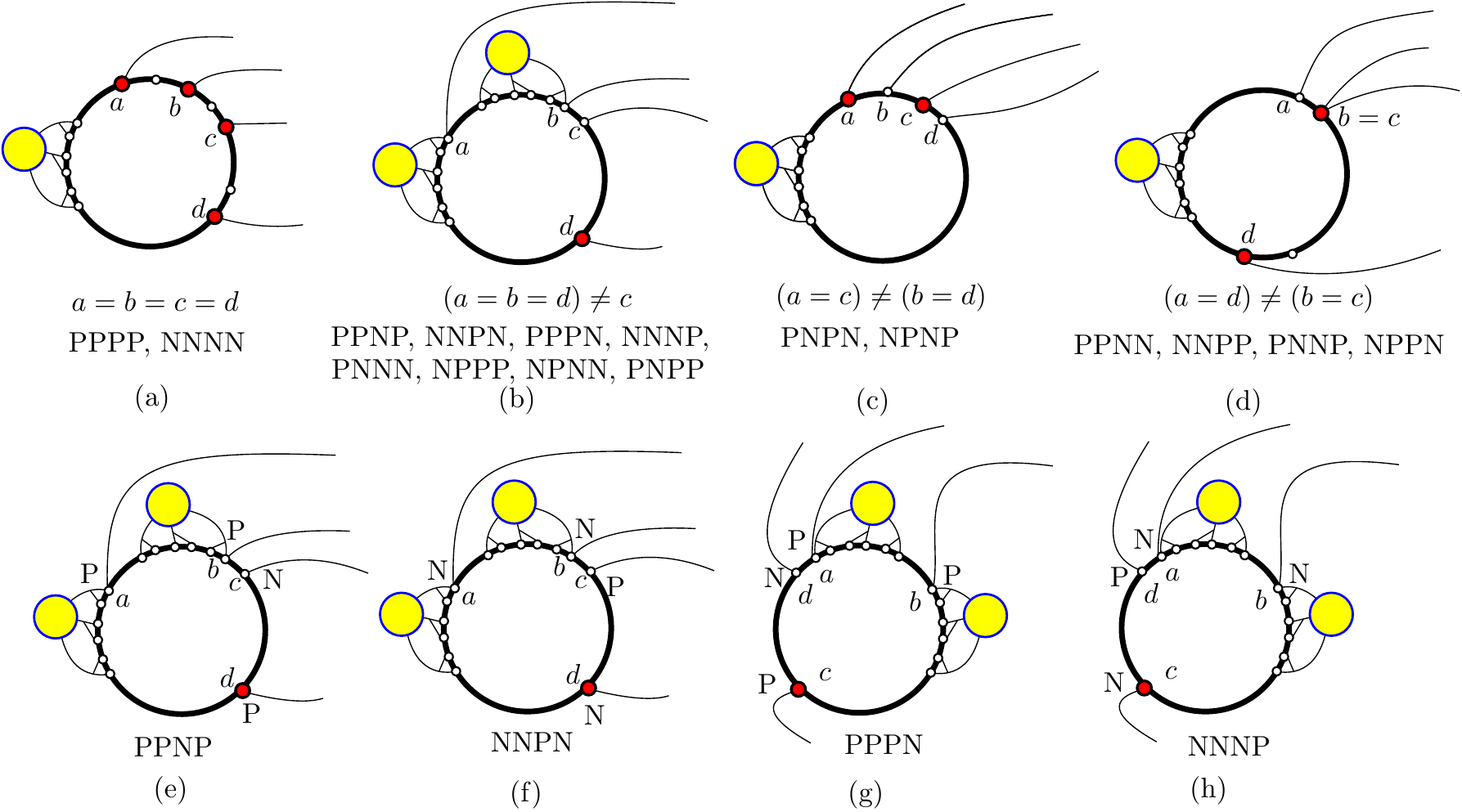}
\caption{Replacing each lamina with a variable gadget, where the degree two vertices on the variable belt are shown in red. (a)--(d) Four configurations based on the 16 scenarios of $a,b,c,d$. (e)--(h) Details for the first four scenarios of Figure~\ref{fig:vg}(b). }\label{fig:vg}
\end{figure}
 
Note that a variable belt can have vertices of degree two, which are shown in large red vertices. The configurations ensure  that the edges incident to the lamina are attached to these vertices. In addition, the degree of each variable vertex does not exceed 4. We now have the following lemma.  
\begin{figure}[h]
\centering 
\includegraphics[width=.75\textwidth]{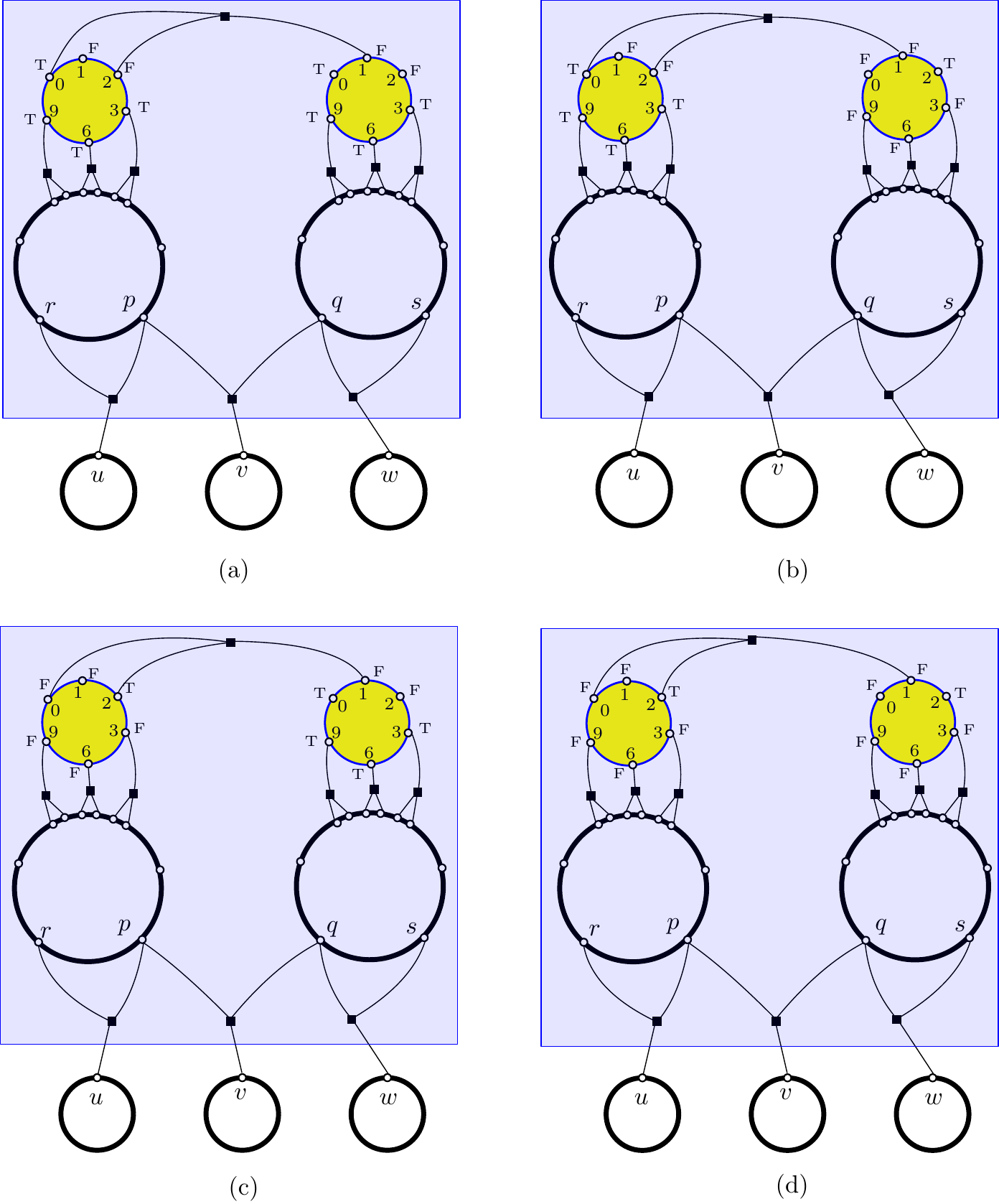}
\caption{Extension of the truth value assignments of $u,v$ and $w$. (a) $u$ and $w$ both obtain  true values. (b) $u$ obtains a true value and $w$ obtains a false value. (c) $u$ obtains a false value and $w$ obtains a true value. (d) $u$ and $w$ both obtain  false values. }\label{fig:aug}
\end{figure}
\begin{lemma}\label{l3}
Let $I$ and $I'$ be the 3-SAT and 1-in-3-SAT instances corresponding to $G$ and $G'$, respectively. Then 
$I$ is satisfiable if and only if $I'$ is satisfiable. 
\end{lemma}
\begin{proof}
The clause gadget ensures that a clause in the  3-SAT instance is satisfiable if and only if each clause  in the clause gadget obtains exactly one true value. The variable gadgets ensure that in every  satisfying truth assignment  $I'$, the truth values of the  positive and negative literals have been set consistently. Since all clauses are satisfied in $I'$, the truth assignment also satisfies $I$. On the other hand, any satisfying truth assignment of $I$ determines a consistent set of truth values on the belt, and thus can be extended to a satisfying truth assignment for $I'$.    
\end{proof}

\subsection{Construction of $G''$}
The 1-in-3-SAT graph $G'$ is not 3-connected. We now add some more clauses and variables to $G$ to construct another 1-in-3-SAT graph $G''$ such that $G''$ is planar 3-connected with the degree of variable vertices bounded by 4. 

Let $r,p,q,s$ be the variables in  a clause gadget of $G'$, e.g., see Figure~\ref{fig:one-in-three}(c). To construct $G''$, we augment the clause gadgets with two modified  vertex gadgets. Note that an outer ring in a variable gadget enforces that the variable belt must obtain an alternating  sequence of true and false values. A \emph{modified vertex gadget} attaches the outer ring to the variable belt at some   positions, as illustrated in  Figure~\ref{fig:aug}(a),  such that the outer ring no longer enforces the variable belt to contain both true and false values. The vertices $r,p$ (similarly, $q,s$) lie on two consecutive vertices on a variable belt. We ensure that $p$ and $q$ are vertices of degree four (i.e., they are degree-two vertices on the variable belt).

\begin{lemma}\label{l4}
Let $I'$ and $I''$ be the 1-in-3-SAT instances  corresponding to $G'$ and $G''$, respectively. Then $I'$ is satisfiable if and only if $I''$ is satisfiable.
\end{lemma}
\begin{proof}
Consider first the case when $I''$ admits a satisfying truth assignment $\phi''$. Since the clauses of $I'$ are included in $I''$, $\phi''$ will set exactly one variable to true in each clause of $I'$. Therefore, $I'$ will be satisfied. 

Assume now that $I'$ admits a satisfying truth assignment $\phi'$. We now show how the truth assignment can be extended to satisfy all the clauses of $I''$. Let $(u\vee p \vee r) \wedge (v\vee p \vee q) \wedge (w\vee q \vee s)$ be three clauses in $I'$ that belong to a clause gadget of $G'$. We now consider the following cases depending on the truth values of $u$ and $w$ in $\phi'$.

\textbf{Case 1 ($u$ and $w$ both obtain true values)}. In this case $r,p,q,s$ must obtain false values. Since $I'$ is satisfied, $v$ must be true. Since the outer ring here no longer enforces the truth value sequence on the variable band to be $T,F,\ldots,T,F$, we can use a truth value sequence $F,F,\ldots, F,F$ for the variable belt.  

The variable belt is connected to the outer ring at the third, sixth and ninth variable vertices on the ring. By the property of the ring (Remark~\ref{r1}), in any satisfying truth  assignment, they must all have the same truth value. Thus the outer rings obtain true values at positions 3, 6 and 9. This truth value assignment is consistent also  for the clause connecting the two rings. Hence we can satisfy all clauses with exactly one true value per clause, as illustrated in Figure~\ref{fig:aug}(a).  

\textbf{Case 2 (exactly one of $u$ and $w$ obtains a true value)}. Without loss of generality assume that $u$ obtains a true value and $w$ obtains a false value (Figure~\ref{fig:aug}(b)). The other case is similar, as shown in Figure~\ref{fig:aug}(c). Since $u$ obtains a true value, $r,p$ must obtain false values and this enforces false values on the corresponding ring. Since $W$ obtains a false value, exactly one of $q,s$ obtains a true value. We choose a feasible truth value assignment based on the value of $v$ in $\phi'$. Hence the corresponding ring obtains false values. We can extend this assignment to also satisfy the clause connecting the two rings. Hence we can satisfy all clauses with exactly one true value per clause.
 
\textbf{Case 3 ($u$ and $w$ both obtain false values)}.
In this case one of $p,r$ obtains a  true  value and the other variable obtains a false value. Similarly, one of $q,s$ obtains a  true  value and the other variable obtains  a false value.  We choose a feasible truth value assignment based on the value of $v$ in $\phi'$. In this scenario, the outer rings obtain false values at positions 3, 6 and 9. This truth value assignment is consistent also  for the clause connecting the two rings. Hence we can satisfy all clauses with exactly one true value per clause, as shown in Figure~\ref{fig:aug}(d). 
\end{proof}

\begin{lemma}\label{property}
The  1-in-3-SAT graph $G''$ is 3-connected with the degree of each variable vertex bounded by 4.
\end{lemma}
\begin{proof}
The input SAT graph $G$ is 3-connected, and by Lemma~\ref{3con}, it remains 3-connected after we replace each variable vertices with a lamina. Each variable gadget is 2-connected and internally 3-connected. Hence by Lemma~\ref{3con}, replacing each lamina with a variable gadget (e.g., Figure~\ref{fig:vg}) keeps the graph 3-connected. Each clause gadget consists of two modified vertex gadgets, e.g., see the subgraph interior the rectangular region in Figure~\ref{fig:aug}(a). Such clause gadgets are 2-connected and internally 3-connected. Hence  by Lemma~\ref{3con}, replacing the clauses with clause gadget keeps the resulting graph $G''$ 3-connected. By construction, the vertices of each variable and clause  gadget are of maximum degree 4. Hence the degree of the vertices of $G''$ are also bounded by 4. 
\end{proof}

\begin{theorem}
\label{th:hard}
Planar 3-connected 1-in-3-SAT is NP-hard even when every variable appears in at most 4 clauses.
\end{theorem}
\begin{proof}
Let $I$ be a planar 3-connected 3-SAT instance and let $G$ be the corresponding SAT graph. Let $G'$ be the  graph  obtained by modifying $G$ using the vertex and clause gadgets, and let $G''$ be the 1-in-3SAT graph obtained from $G'$. Let $I'$ and $I''$ be the 1-in-3-SAT
 instances corresponding to $G'$ and $G''$, respectively. By Lemma~\ref{l3} and Lemma~\ref{l4}, $I$ is satisfiable if and only if $I''$ is satisfiable.  By Lemma~\ref{property}, $G''$ is a 3-connected  graph with the degree of its variable vertices bounded by four. 
     
\end{proof}

\section{Positive Planar $3$-Connected $1$-in-$3$-SAT with Even Variable Frequency}
\label{sec:even}
In this section we prove that every positive planar 3-connected 1-in-3-SAT with each variable appearing in an even number of clauses is always satisfiable and a satisfying truth assignment can be computed in quadratic time. 

\begin{figure}[h]
\centering 
\includegraphics[width=.75\textwidth]{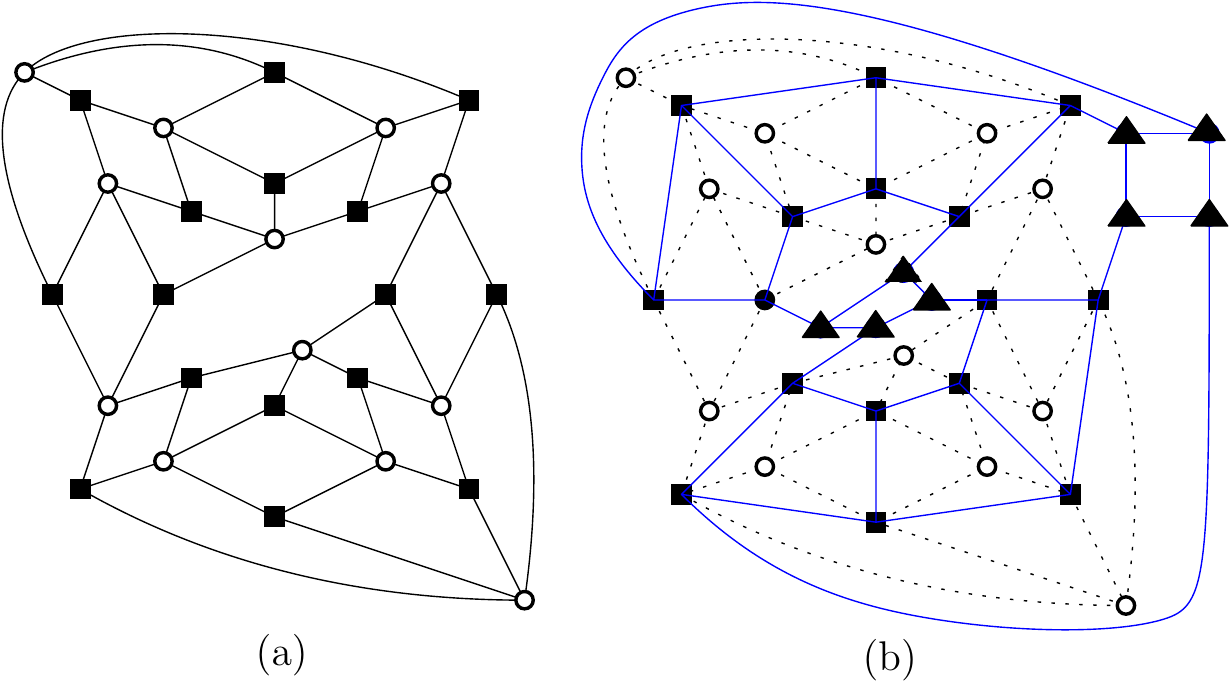}
\caption{(a) An instance of a positive planar 3-connected $1$-in $3$-SAT, and (b) the corresponding saturated clause graph shown in solid blue.}\label{fig:one-in-threex}
\end{figure}

Let $R$ be an arbitrary positive planar $3$-connected $1$-in-$3$-SAT expression and let $G$ be its corresponding SAT graph. Figure~\ref{fig:one-in-threex}(a) illustrates such an instance where the clause and variable vertices are shown in squares and circles, respectively. 

We first construct a saturated clause graph $C_s$, as illustrated in  Figure~\ref{fig:one-in-threex}(b). The vertices of the added cycle are shown in triangles. We have shown in the proof of Theorem~\ref{th:nae3sat} that $C_s$ is a bridgeless cubic graph. However, here we need to show that $C_s$ is 3-connected. To observe this we can think of the construction of $C_s$ in three steps where each step remains the graph 3-connected, as follows.

\begin{enumerate}[]
    \item Step 1: For each quadrangular face of the clause graph $C$, add an edge between the clause vertices. Let $C'$ be the resulting graph. Since $C$ is 3-connected and $C'$ is obtained by adding some edges to $C$, $C'$ is 3-connected. 
    \item Step 2: Apply the saturation operation to all the faces of length more than four. Let $C''$ be the resulting graph. By Lemma~\ref{3con}, $C''$ is 3-connected.
    \item Step 3: Delete the variable vertices from $C''$ to obtain $C_s$. If the neighbors of a vertex $v$ in a 3-connected graph induces a cycle, then $v$  can be deleted to obtain another 3-connected graph~\cite{Tutte}. Using this property one can show that $C_s$ is 3-connected.
\end{enumerate}



Since $C_s$ is a planar 3-connected   cubic graph, the dual of $C_s^*$ is a triangulated planar graph. Since every variable appears in an even number of clauses, each vertex of $C_s^*$ has an even degree. Every triangulated planar graph where each vertex has an even degree admits a 3 vertex coloring~\cite{tsai2011new}. Such a planar  triangulation is known as an \emph{even triangulation}.   Therefore, we can color the faces of $C_s $ with 3 colors $c_1,c_2,c_3$ such that no two adjacent faces receive the same color.  Figure~\ref{fig:one-in-three2}(a) illustrates such a coloring of the faces.  We take the  faces that are colored with $c_1$ and set the variables corresponding to those faces to true. Figure~\ref{fig:one-in-three2}(b) illustrates these variables in   filled circles (orange).  Finally, we set the remaining variables to false.  

We now show that the resulting truth assignment satisfies the 1-in 3-SAT instance $R$. Since the SAT graph $G$ is 3-connected, every clause is adjacent to exactly three variables. Therefore, the three faces around a clause vertex must have 3 different colors. One of these is colored with $c_1$ and the remaining two must be colored with $c_2$ and $c_3$. Therefore, exactly one variable associated to that clause is set to true. 

The construction of $C_s$ and the dual graph takes linear time. A 3-coloring of an even triangulation can be computed using a concept of `non-crossing Eulerian circuit'~\cite{tsai2011new}, which is straightforward to compute in quadratic    time. Hence in quadratic time, one can find a satisfying truth assignment for $R$.

\begin{figure}[h]
\centering 
\includegraphics[width=.7\textwidth]{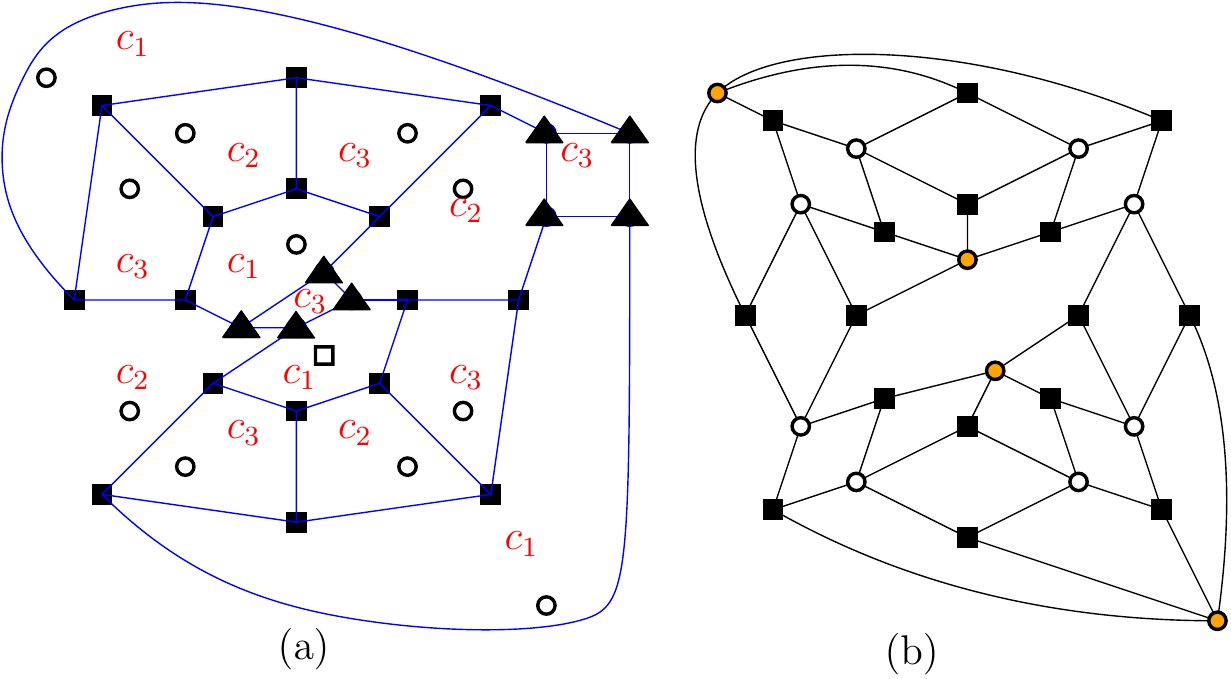}
\caption{Construction of a satisfying truth assignment from a face coloring of $C_s$. }\label{fig:one-in-three2}
\end{figure}
The following theorem summarizes the results of this section.

\begin{theorem}
\label{th:pln}
Let $R$ be an  arbitrary positive planar $3$-connected 1-in-$3$-SAT expression. If every variable appears in an even number clauses, then $R$ is always satisfiable and a satisfiable assignment of $R$ can be computed in quadratic time. 
\end{theorem}

\section{Conclusions}
\label{sec:con}

We have shown that positive planar $3$-connected NAE $3$-SAT is always satisfiable and a satisfiable assignment can be obtained in linear time. We have also shown that deciding whether a positive planar NAE $3$-SAT is satisfiable or not is as hard as finding a spanning 2-matching in a cubic graph. Hence we pose the following open problems.

\bigskip\noindent
\textbf{Open Problem 1.} Does there exist a  linear-time algorithm for finding a spanning 2-matching in a cubic (or, subcubic) planar graph?
\bigskip

 It would be interesting to examine whether a linear-time algorithm can be obtained for the  general case. In fact, we are not aware of any  $o(n^{3/2})$-time algorithm for the general case.\bigskip
 
\noindent
\textbf{Open Problem 2.} Does there exist  an $o(n^{3/2})$-time algorithm for solving positive planar NAE 3-SAT? 
\bigskip


One can also try to find out other tractable versions for $3$-SATs based on connectivity constraints on the SAT graph. We have shown that planar 3-connected 1-in-3-SAT remains NP-hard when every variable appears in at most 4 clauses.  Thus the following is another intriguing question. 
\bigskip

\noindent
\textbf{Open Problem 3.} Does  planar 3-connected 1-in-3-SAT remain NP-hard when every variable appears in  3 clauses? 
\bigskip

 We have shown that if every variable appears in an even number clauses, then a positive planar $3$-connected 1-in-$3$-SAT is always satisfiable and a satisfiable assignment can be computed in quadratic time. The time complexity relies on computing a  3-coloring of an even triangulation. Hence it would be interesting to examine whether there exist faster algorithms to compute a 3-coloring for an even triangulation.

\subsubsection*{Acknowledgement}

The work of D. Mondal is supported by the Natural Sciences and Engineering Research Council of Canada (NSERC).

\bibliographystyle{abbrv}
\bibliography{bibl} 

\end{document}